\documentclass[11pt]{article}

\usepackage[T1]{fontenc}
\usepackage[utf8]{inputenc}
\usepackage{lmodern}

\usepackage{microtype}

\usepackage{amsmath,amssymb,amsthm,mathtools}

\theoremstyle{plain}
\newtheorem{theorem}{Theorem}[section]
\newtheorem{lemma}[theorem]{Lemma}

\newtheorem{corollary}[theorem]{Corollary}

\theoremstyle{definition}

\theoremstyle{remark}
\newtheorem{remark}[theorem]{Remark}

\DeclareMathOperator{\dom}{dom}

\newcommand{\SAT}{\ensuremath{\textnormal{\textsc{3SAT}}}}
\newcommand{\DNP}{\ensuremath{\mathrm{DNP}}}
\newcommand{\DP}{\ensuremath{\mathrm{DP}}}
\newcommand{\NP}{\ensuremath{\mathrm{NP}}}
\newcommand{\coNP}{\ensuremath{\mathrm{coNP}}}

\newcommand{\ExactDNP}[1]{\ensuremath{\mathrm{Exact}\text{-}#1\text{-}\DNP}}
\newcommand{\SATUNSAT}{\ensuremath{\SAT\text{-}\mathrm{UNSAT}}}

\usepackage{xcolor}
\usepackage[
  colorlinks=true,
  linkcolor=blue,
  citecolor=blue,
  urlcolor=blue
]{hyperref}

\usepackage[nameinlink,noabbrev]{cleveref}
\crefname{theorem}{Theorem}{Theorems}
\Crefname{theorem}{Theorem}{Theorems}
\crefname{lemma}{Lemma}{Lemmas}
\Crefname{lemma}{Lemma}{Lemmas}
\crefname{corollary}{Corollary}{Corollaries}
\Crefname{corollary}{Corollary}{Corollaries}
\crefname{section}{Section}{Sections}
\Crefname{section}{Section}{Sections}
\crefname{remark}{Remark}{Remarks}
\Crefname{remark}{Remark}{Remarks}

\title{Closing the Complexity Gap for Exact Domatic Number at Three and Four}
\title{Closing the Complexity Gap for Exact Domatic Number at Three and Four}

\author{
Holger Spakowski\\
Department of Mathematics and Applied Mathematics\\
University of Cape Town\\
Rondebosch 7701, South Africa\\
\texttt{Holger.Spakowski@uct.ac.za}
}

\date{}

\begin{document}

\maketitle

\begin{abstract}
The exact domatic-number problem asks, for a fixed integer \(k\), whether a
given graph \(G\) satisfies \(\dom(G)=k\). Riege and Rothe proved
\(\DP\)-completeness for every fixed \(k\ge 5\), while the cases \(k=3\)
and \(k=4\) remained open. We close this classification gap. The main
ingredient is a polynomial-time reduction from \(\SAT\) whose output
graphs have domatic number \(4\) in the satisfiable case and domatic number
\(2\) in the unsatisfiable case; in particular, the reduction never produces
a graph of domatic number \(3\). This directly realizes the route suggested
by Riege and Rothe for closing the remaining cases. Together with a simpler
three-versus-two reduction, this yields \(\DP\)-completeness of
\(\ExactDNP{3}\) and \(\ExactDNP{4}\). The proofs are constructive and give
explicit graph gadgets whose local domination constraints encode truth
assignments and clause satisfaction. The soundness arguments show conversely
that any sufficiently large domatic partition enforces the intended consistency
conditions and therefore yields a satisfying assignment. Consequently,
\(\ExactDNP{k}\) is \(\DP\)-complete for every fixed \(k\ge 3\), completing
the fixed-value classification from \(k=3\) onward.
\end{abstract}

\begin{quote}
\end{quote}

\section{Introduction}
\label{sec:introduction}

A dominating set in a graph \(G\) is a set of vertices that reaches every
vertex of \(G\) within distance at most one.  The domatic number problem asks
for a partition of the vertex set into as many disjoint dominating sets as
possible.  The maximum number of parts in such a partition is the domatic
number of \(G\), denoted here by \(\dom(G)\).
The problem was already discussed in the early work of Cockayne and
Hedetniemi on domination in graphs
\cite{CockayneHedetniemi1975,CockayneHedetniemi1977},
and it has since become a standard graph
partitioning problem.  Its interpretation is particularly natural in network
settings: a dominating set may model a set of facilities or transmitting
stations that can serve all vertices of the network,
and a domatic partition then corresponds to several disjoint such service
layers \cite{CockayneHedetniemi1977,RiegeRothe2004ExactDNP}.

For fixed \(k\), the decision problem \(k\)-DNP asks whether
\(\dom(G)\geq k\).
It is known that \(k\)-DNP is NP-complete for every fixed \(k\geq 3\),
whereas \(2\)-DNP is polynomial-time decidable
\cite{GareyJohnson1979,KaplanShamir1994}.
In this paper we study the corresponding exact problems.  For fixed \(k\),
let
\[
  \ExactDNP{k}
  =
  \{G : \dom(G)=k\}.
\]
Exact versions of NP-optimization problems often naturally lie not merely in
NP, but in the class \DP, the second level of the boolean hierarchy over NP:
one has to certify that the optimum is at least a given value and also that it
is not larger than that value.  This viewpoint goes back to the work of
Papadimitriou and Yannakakis on DP, to the subsequent development of the
boolean hierarchy by Cai et al., and to Wagner's technique for proving
hardness in the levels of the boolean hierarchy
\cite{PapadimitriouYannakakis1984,CaiEtAl1988,CaiEtAl1989,Wagner1987}.

Riege and Rothe initiated the systematic study of exact domatic-number
problems within this framework. They proved that Exact-\(i\)-DNP is
\DP-complete for every fixed \(i\geq 5\), while Exact-\(2\)-DNP is
\coNP-complete \cite{RiegeRothe2004ExactDNP}.
Moreover, \(\ExactDNP{1}\) is polynomial-time
decidable, since for every nonempty graph \(\dom(G)=1\) if and only if
\(G\) has an isolated vertex.  Hence the only fixed exact values not covered
by the previous classification were
\[
  i=3
  \qquad\text{and}\qquad
  i=4,
\]
which were left open in the original work of Riege and Rothe
\cite{RiegeRothe2004ExactDNP}.
The same gap is also recorded in later
expositions~\cite{RiegeRothe2006Survey,RotheBook}.
This gap is analogous in spirit to the earlier gap in exact graph
colorability: Wagner proved \DP-completeness for sufficiently large exact
chromatic numbers, and Rothe later proved the optimal threshold by showing
that \(\mathrm{Exact}\text{-}4\text{-Colorability}\) is \DP-complete
\cite{Wagner1987,Rothe2003ExactFourColorability}.  Rothe's proof
uses the reduction of Guruswami and Khanna, whose crucial feature for this
application is that the constructed graphs avoid chromatic number \(4\) in
one of the two cases
\cite{GuruswamiKhanna2000,Rothe2003ExactFourColorability}.

Riege and Rothe observed that an analogous way to close the remaining
exact domatic-number gap would be to find a reduction from a suitable
\NP-complete problem to the domatic-number problem whose output graphs
never have domatic number \(3\)
\cite{RiegeRothe2004ExactDNP}.
The
main reduction of the present paper supplies such a construction.  It maps a
formula \(\varphi\) to a graph \(R(\varphi)\) such that
\[
  \varphi\in\SAT
  \Longrightarrow
  \dom(R(\varphi))=4,
  \qquad
  \varphi\notin\SAT
  \Longrightarrow
  \dom(R(\varphi))=2.
\]
Thus the image of the reduction avoids domatic number \(3\).  This is the
technical ingredient that closes the open cases.

Our main contribution is therefore the following theorem.

\begin{theorem}[Main theorem]
\label{thm:main-result}
  \(\ExactDNP{3}\) and
  \(\ExactDNP{4}\) are \DP-complete.
\end{theorem}

Together with the results of Riege and Rothe, this yields the following
complete classification for all fixed exact values at least three.

\begin{corollary}
\label[corollary]{cor:exact-k-dnp-classification}
  For every fixed integer \(k\geq 3\),
  \(\ExactDNP{k}\) is \DP-complete.
\end{corollary}

The proof proceeds through two direct SAT-to-domatic-number reductions.
The first reduction maps a formula \(\varphi\) to a graph \(S(\varphi)\) with
\[
  \varphi\in\SAT
  \Longrightarrow
  \dom(S(\varphi))=3,
  \qquad
  \varphi\notin\SAT
  \Longrightarrow
  \dom(S(\varphi))=2.
\]
The second, stronger reduction is the avoidance reduction described above:
it maps \(\varphi\) to \(R(\varphi)\) with
\[
  \varphi\in\SAT
  \Longrightarrow
  \dom(R(\varphi))=4,
  \qquad
  \varphi\notin\SAT
  \Longrightarrow
  \dom(R(\varphi))=2.
\]
The technical sections establish these two SAT-to-DNP reductions directly.  The final
application section combines them with the canonical \DP-complete problem
\(\SATUNSAT\) \cite{PapadimitriouYannakakis1984}, where
\[
  \SATUNSAT
  =
  \{(\varphi,\psi):\varphi\in\SAT
  \text{ and } \psi\notin\SAT\},
\]
and with elementary graph operations, namely disjoint union and join.

\newpage

\subsection{Related Work}

The domatic number problem belongs to the broader theory of domination in
graphs, whose origins include the work of Cockayne and Hedetniemi
\cite{CockayneHedetniemi1977}.  The classical NP-completeness background is
given in the compendium of Garey and Johnson, which lists the domatic number
problem and attributes the reduction from \(\SAT\) to unpublished work of
Garey, Johnson, and Tarjan \cite{GareyJohnson1979}.  Kaplan and Shamir gave
an explicit reduction from graph colorability to the domatic number problem
with useful structural properties; this reduction is also used as a point of
departure in the work of Riege and Rothe
\cite{KaplanShamir1994,RiegeRothe2004ExactDNP}.

The complexity-theoretic setting of this paper is the class \DP{} and the boolean
hierarchy over NP.  Papadimitriou and Yannakakis introduced DP in their study
of problems that combine NP-type and coNP-type requirements, including exact
and critical problems \cite{PapadimitriouYannakakis1984}.  Cai et al.
introduced and studied the boolean hierarchy over NP
\cite{CaiEtAl1988,CaiEtAl1989}.  Wagner developed general sufficient
conditions for proving hardness in the levels of this hierarchy and applied
them to exact optimization problems, including exact graph colorability
\cite{Wagner1987}.

Rothe's \DP-completeness result for \(\mathrm{Exact}\text{-}4\text{-Colorability}\) solved the
analogous small-threshold question for chromatic number
\cite{Rothe2003ExactFourColorability,RotheBook}.  The proof relies on the
Guruswami--Khanna reduction, originally developed in connection with the
hardness of coloring 3-colorable graphs with four colors
\cite{GuruswamiKhanna2000,Rothe2003ExactFourColorability,RiegeRothe2006Survey}.
In the exact-colorability application, the relevant point is that this
reduction produces only the two values \(3\) and \(5\), and hence avoids the
troublesome intermediate exact value \(4\).

Riege and Rothe proved the central previously known results for exact
domatic number: \DP-completeness of Exact-\(i\)-DNP for \(i\geq 5\),
\coNP-completeness of Exact-\(2\)-DNP, and corresponding completeness
results for higher levels of the boolean hierarchy and for related generalized
domination problems \cite{RiegeRothe2004ExactDNP}.  They also identified
a sufficient route to the remaining cases: it would be enough to construct a
reduction from a suitable \NP-complete problem to the domatic-number problem
whose outputs have domatic number different from \(3\)
\cite{RiegeRothe2004ExactDNP}.
The same gap is also discussed in the survey article by Riege and
Rothe~\cite{RiegeRothe2006Survey} and in Rothe's textbook~\cite{RotheBook};
for a detailed thesis-level treatment, see Riege~\cite{Riege2006Thesis}.

\section{Preliminaries}
\label{sec:preliminaries}

All graphs in this paper are finite, simple, undirected graphs.  For a graph
\(G\), we write \(V(G)\) and \(E(G)\) for its vertex set and edge set.  For
\(v\in V(G)\), let
\[
  N_G(v)=\{u\in V(G): uv\in E(G)\}
\]
be the open neighborhood of \(v\), and let
\[
  N_G[v]=N_G(v)\cup\{v\}
\]
be the closed neighborhood of \(v\).  When the graph is clear from context,
we write \(N(v)\) and \(N[v]\).  The degree of \(v\) is denoted
\(\deg_G(v)\), or simply \(\deg(v)\), and
\[
  \min\deg(G)=\min_{v\in V(G)} \deg_G(v).
\]

A set \(D\subseteq V(G)\) is a dominating set of \(G\) if
\[
  N[v]\cap D\neq\emptyset
  \qquad\text{for every }v\in V(G).
\]
Equivalently, every vertex is either in \(D\) or adjacent to a vertex in
\(D\).  The domatic number of \(G\) is
\[
  \dom(G)
  =
  \max\{k : V(G)\text{ can be partitioned into }k
  \text{ dominating sets}\}.
\]
A partition of \(V(G)\) into \(k\) dominating sets is called a
\(k\)-domatic partition.

It is often useful to view a \(k\)-domatic partition as a coloring.  A map
\[
  c:V(G)\to\{1,\ldots,k\}
\]
is a \(k\)-domatic coloring if every color class
\(c^{-1}(i)\) is a dominating set.  Equivalently,
\[
  c(N[v])=\{1,\ldots,k\}
  \qquad\text{for every }v\in V(G).
\]
Thus every closed neighborhood must see all \(k\) colors.  This is not a
proper graph coloring: adjacent vertices may receive the same color.

We will use the following simple monotonicity observation throughout.  If
\(G\) has a \(k\)-domatic partition and \(1\leq \ell\leq k\), then \(G\) has
an \(\ell\)-domatic partition.  Indeed, in the partition view one obtains an
\(\ell\)-partition by repeatedly merging two parts; supersets of dominating
sets are again dominating.  Equivalently, in the coloring view one obtains
an \(\ell\)-domatic coloring by identifying colors.  Every closed
neighborhood that saw all old colors still sees all remaining colors after
such identifications.

We write \(\SAT\) for the standard NP-complete language of satisfiable
Boolean formulas in conjunctive normal form with exactly three literal
occurrences per clause.  Thus, for such a formula \(\varphi\), the assertion
\(\varphi\in\SAT\) means that \(\varphi\) has a satisfying truth assignment.

For fixed \(k\geq 1\), define
\[
  k\text{-DNP}
  =
  \{G : \dom(G)\geq k\},
\]
and define the exact version by
\[
  \ExactDNP{k}
  =
  \{G : \dom(G)=k\}.
\]

We recall the complexity class \DP{} in the form needed here.  For classes
\(\mathcal C\) and \(\mathcal D\), their complex intersection is
\[
  \mathcal C\wedge \mathcal D
  =
  \{A\cap B : A\in\mathcal C,\ B\in\mathcal D\}.
\]
Then
\[
  \DP
  =
  \NP\wedge \coNP.
\]
Equivalently, a language is in DP if it can be written as the intersection
of a language in NP and a language in coNP
\cite{PapadimitriouYannakakis1984,CaiEtAl1988,CaiEtAl1989}.

The following membership observation is standard.  It is the domatic-number
instance of the usual \DP{} upper bound for exact optimization problems: one asks
simultaneously for a lower bound and for the failure of the next larger lower
bound.  We include the proof for completeness and to fix notation; see also
\cite{PapadimitriouYannakakis1984,RiegeRothe2004ExactDNP}.

\begin{lemma}[Standard membership]
\label[lemma]{lem:exact-dnp-in-dp}
  For every fixed \(k\geq 1\),
  \[
    \ExactDNP{k}\in \DP.
  \]
\end{lemma}

\begin{proof}
  For fixed \(k\), the language \(k\)-DNP is in NP: one guesses a partition
  of \(V(G)\) into \(k\) parts and verifies in polynomial time that every
  part is a dominating set.  Hence \((k+1)\)-DNP is also in NP, and its
  complement is in coNP.

  Now
  \[
    \dom(G)=k
    \Longleftrightarrow
    \dom(G)\geq k
    \ \wedge\
    \neg(\dom(G)\geq k+1).
  \]
  Therefore
  \[
    \ExactDNP{k}
    =
    k\text{-DNP}
    \cap
    \overline{(k+1)\text{-DNP}},
  \]
  which is the intersection of an NP language and a coNP language.  Thus
  \(\ExactDNP{k}\in\DP\).
\end{proof}

The following elementary facts about domatic partitions will be used
repeatedly.  They are standard; we include the short proofs to make the paper
self-contained.  We prove these facts in the
coloring view of domatic partitions.

For disjoint graphs \(G\) and \(H\), their disjoint union is denoted
\(G\dot\cup H\).  Thus
\[
  V(G\dot\cup H)=V(G)\cup V(H),
  \qquad
  E(G\dot\cup H)=E(G)\cup E(H).
\]
Their join is denoted \(G\oplus H\).  It is obtained from \(G\dot\cup H\)
by adding every edge between \(V(G)\) and \(V(H)\).  We write \(K_t\) for the
clique on \(t\) vertices, with the convention that \(K_0\) is the empty
graph.

\begin{lemma}[Disjoint unions]
\label[lemma]{lem:disjoint-union}
  For all disjoint graphs \(G\) and \(H\),
  \[
    \dom(G\dot\cup H)=\min\{\dom(G),\dom(H)\}.
  \]
\end{lemma}

\begin{proof}
  Let \(r=\dom(G\dot\cup H)\), and let
  \(c:V(G\dot\cup H)\to\{1,\ldots,r\}\) be an \(r\)-domatic coloring.
  Since there are no edges between \(G\) and \(H\), every closed
  neighborhood of a vertex of \(G\) in \(G\dot\cup H\) is just its closed
  neighborhood in \(G\).  Hence the restriction of \(c\) to \(V(G)\) is an
  \(r\)-domatic coloring of \(G\).  Similarly, the restriction of \(c\) to
  \(V(H)\) is an \(r\)-domatic coloring of \(H\).  Thus
  \[
    r\leq \dom(G)
    \qquad\text{and}\qquad
    r\leq \dom(H),
  \]
  so
  \[
    \dom(G\dot\cup H)\leq \min\{\dom(G),\dom(H)\}.
  \]

  Conversely, let \(d=\min\{\dom(G),\dom(H)\}\).  By monotonicity, \(G\)
and \(H\) have \(d\)-domatic colorings with the same color set
\(\{1,\ldots,d\}\).  Color \(G\dot\cup H\) by using these two colorings
on the two components.  If \(v\in V(G)\), then
\[
  N_{G\dot\cup H}[v]=N_G[v],
\]
and hence \(N_{G\dot\cup H}[v]\) sees all \(d\) colors.  Similarly, if
\(v\in V(H)\), then
\[
  N_{G\dot\cup H}[v]=N_H[v],
\]
and again the closed neighborhood sees all \(d\) colors.  Thus the
combined coloring is a \(d\)-domatic coloring of \(G\dot\cup H\), and so
\[
  \dom(G\dot\cup H)\geq d.
\]
Together with the already proved inequality
\[
  \dom(G\dot\cup H)\leq \min\{\dom(G),\dom(H)\}=d,
\]
the equality follows.
\end{proof}

\begin{lemma}[Joining a clique]
\label[lemma]{lem:join-clique}
  For every graph \(G\) and every \(t\geq 0\),
  \[
    \dom(G\oplus K_t)=\dom(G)+t.
  \]
\end{lemma}

\begin{proof}
  The case \(t=0\) is immediate, so assume \(t\geq 1\).  Let
  \(d=\dom(G)\), and let
  \(c:V(G)\to\{1,\ldots,d\}\) be a \(d\)-domatic coloring of \(G\).  Write
  \(V(K_t)=\{x_1,\ldots,x_t\}\).  Extend \(c\) to \(G\oplus K_t\) by giving
  \(x_j\) the new color \(d+j\), for \(1\leq j\leq t\).

  Every vertex of \(G\) sees all colors \(1,\ldots,d\) inside its closed
  neighborhood in \(G\), and it sees all new colors \(d+1,\ldots,d+t\)
  because it is adjacent to every vertex of \(K_t\).  Every vertex of
  \(K_t\) has the whole graph \(G\oplus K_t\) as its closed neighborhood.
  Hence this is a \((d+t)\)-domatic coloring, and so
  \[
    \dom(G\oplus K_t)\geq \dom(G)+t.
  \]

For the reverse inequality, let
\(r=\dom(G\oplus K_t)\), and let
\[
  c:V(G\oplus K_t)\to\{1,\ldots,r\}
\]
be an \(r\)-domatic coloring of \(G\oplus K_t\).  Thus
\(\{1,\ldots,r\}\) is the set of colors used by \(c\).  Let \(T\) be the
set of colors that appear on the clique \(K_t\).  Then \(|T|\leq t\).  Let
\[
  S=\{1,\ldots,r\}\setminus T
\]
be the set of remaining colors, that is, the colors absent from \(K_t\).

If \(S=\emptyset\), then all \(r\) colors appear on the \(t\)-vertex clique
\(K_t\).  Hence
\[
  r=|T|\leq t\leq t+\dom(G),
\]
so the desired upper bound is immediate.  Thus we may assume
\(S\neq\emptyset\).
We construct an
\(|S|\)-domatic coloring of \(G\).  Keep all vertices of \(G\) whose color
lies in \(S\) unchanged, and recolor every vertex of \(G\) whose color lies
in \(T\) arbitrarily with one fixed color from \(S\).

We claim that the resulting coloring of \(G\) is \(|S|\)-domatic.  Fix
\(v\in V(G)\) and \(\alpha\in S\).  Since \(c\) is a domatic coloring of
\(G\oplus K_t\), the closed neighborhood of \(v\) in \(G\oplus K_t\)
contains a vertex of color \(\alpha\).  But \(\alpha\notin T\), so color
\(\alpha\) does not occur on \(K_t\).  Hence this vertex of color
\(\alpha\) must already lie in \(N_G[v]\).  It is not recolored, since its
color lies in \(S\).  Thus, after the recoloring, \(N_G[v]\) still sees
color \(\alpha\).  Since this holds for every \(v\in V(G)\) and every
\(\alpha\in S\), we obtain an \(|S|\)-domatic coloring of \(G\).  Therefore
\[
  |S|\leq \dom(G).
\]
Since \(r=|T|+|S|\) and \(|T|\leq t\), we have
\[
  r\leq t+\dom(G).
\]
Recalling that \(r=\dom(G\oplus K_t)\), this proves
\[
  \dom(G\oplus K_t)\leq \dom(G)+t.
\]
Together with the opposite inequality proved above, this gives
\[
  \dom(G\oplus K_t)=\dom(G)+t.
\]
\end{proof}

\begin{lemma}[Degree upper bound]
\label[lemma]{lem:degree-upper-bound}
  For every graph \(G\),
  \[
    \dom(G)\leq \min\deg(G)+1.
  \]
\end{lemma}

\begin{proof}
  Let \(k=\dom(G)\), and let \(c\) be a \(k\)-domatic coloring of \(G\).
  For every vertex \(v\in V(G)\), the closed neighborhood \(N[v]\) must see
  all \(k\) colors.  Therefore
  \[
    k\leq |N[v]|=\deg(v)+1.
  \]
  Taking \(v\) of minimum degree gives
  \[
    \dom(G)=k\leq \min\deg(G)+1.
  \]
\end{proof}

\begin{lemma}[Graphs without isolated vertices]
\label[lemma]{lem:no-isolated-vertices}
  If \(G\) has no isolated vertices, then
  \[
    \dom(G)\geq 2.
  \]
\end{lemma}

\begin{proof}
  Let \(I\) be a maximal independent set of \(G\).  Define a coloring
  \(c:V(G)\to\{1,2\}\) by coloring the vertices of \(I\) with color \(1\)
  and the vertices of \(V(G)\setminus I\) with color \(2\).

  We show that every closed neighborhood sees both colors.  If \(v\in I\),
  then \(v\) itself has color \(1\).  Since \(G\) has no isolated vertices,
  \(v\) has a neighbor; and since \(I\) is independent, every neighbor of
  \(v\) lies in \(V(G)\setminus I\), so \(v\) also sees color \(2\).

  If \(v\in V(G)\setminus I\), then \(v\) itself has color \(2\).  By
  maximality of \(I\), the vertex \(v\) has a neighbor in \(I\), and
  therefore \(v\) also sees color \(1\).  Hence \(c\) is a \(2\)-domatic
  coloring of \(G\), and so \(\dom(G)\geq 2\).
\end{proof}

\section{A Three-versus-Two SAT-to-DNP Reduction}
\label{sec:warm-up-reduction}

We first present a direct reduction from \SAT{} to the domatic
number problem.  This reduction is not intended to be edge-minimal.  Its
purpose is to isolate the local mechanisms that will also be used in the
main construction.  A large clique makes the anchor and literal vertices
easy to dominate once all three colors occur in the clique.  The actual
logical constraints are imposed by small degree-two helper vertices.  In
the coloring view of domatic partitions, a \(3\)-domatic coloring is a
coloring of the vertices with three colors such that every closed
neighborhood sees all three colors.

\begin{theorem}[Warm-up three-versus-two reduction]
\label{thm:warm-up-gap}
There is a polynomial-time computable function \(S\) that maps every
\SAT{} instance \(\varphi\) to a graph
\[
  G_\varphi := S(\varphi)
\]
such that
\[
  \varphi \in \SAT
  \quad\Longrightarrow\quad
  \dom(G_\varphi)=3,
\]
and
\[
  \varphi \notin \SAT
  \quad\Longrightarrow\quad
  \dom(G_\varphi)=2.
\]
Equivalently,
\[
  \dom(S(\varphi))
  =
  \begin{cases}
    3, & \text{if } \varphi \in \SAT,\\
    2, & \text{if } \varphi \notin \SAT.
  \end{cases}
\]
\end{theorem}

\begin{proof}
Let
\[
  \varphi = C_1 \wedge \cdots \wedge C_m
\]
be a 3-CNF formula over variables \(x_1,\ldots,x_n\), where
\[
  C_j = (\ell_{j,1}\vee \ell_{j,2}\vee \ell_{j,3})
\]
for \(1\leq j\leq m\).  We construct a graph \(G_\varphi\).

\paragraph{Vertex set.}
The graph \(G_\varphi\) has the following vertices.

First, it has three anchor vertices
\[
  \tau,\ f,\ b.
\]
The vertex \(\tau\) will be the true-color anchor.  The vertices \(f\)
and \(b\) serve as anchors for the other two colors.

For each variable \(x_i\), we create two literal vertices
\[
  x_i,\qquad \bar{x}_i.
\]
For each variable \(x_i\), we also create one variable-helper vertex
\[
  h_i.
\]
For each clause \(C_j\), we create one clause vertex
\[
  c_j
\]
and one clause-helper vertex
\[
  q_j.
\]

Thus
\[
\begin{aligned}
  V(G_\varphi)
  ={}& \{\tau,f,b\}
      \cup \{x_i,\bar{x}_i : 1\leq i\leq n\} \\
     &{}\cup \{h_i : 1\leq i\leq n\}
      \cup \{c_j,q_j : 1\leq j\leq m\}.
\end{aligned}
\]

\paragraph{The big clique.}
Put precisely the vertices
\[
  K
  :=
  \{\tau,f,b\}
  \cup
  \{x_i,\bar{x}_i : 1\leq i\leq n\}
\]
into one clique.  Thus the clique vertices are exactly the anchors and
the literal vertices.  The vertices not placed in this clique are the
variable helpers \(h_i\), the clause vertices \(c_j\), and the clause
helpers \(q_j\).

\paragraph{Variable consistency helpers.}
For each variable \(x_i\), add exactly the two edges
\[
  \{h_i,x_i\},\qquad \{h_i,\bar{x}_i\}.
\]
Hence
\[
  N[h_i]=\{h_i,x_i,\bar{x}_i\}.
\]
In any \(3\)-domatic coloring, the closed neighborhood \(N[h_i]\) must
see all three colors.  Since \(N[h_i]\) has exactly three vertices, the
vertices
\[
  h_i,\quad x_i,\quad \bar{x}_i
\]
must receive three distinct colors.  In particular, the two literal
vertices \(x_i\) and \(\bar{x}_i\) cannot both receive the color of
\(\tau\).

\paragraph{Clause helpers.}
For each clause \(C_j\), add exactly the two edges
\[
  \{q_j,\tau\},\qquad \{q_j,c_j\}.
\]
Hence
\[
  N[q_j]=\{q_j,\tau,c_j\}.
\]
In any \(3\)-domatic coloring, if \(\tau\) has color \(T\), then
\(q_j\) and \(c_j\) must receive the two colors different from \(T\).
Thus neither \(q_j\) nor \(c_j\) can have color \(T\).

\paragraph{Clause vertices.}
For each clause
\[
  C_j=(\ell_{j,1}\vee \ell_{j,2}\vee \ell_{j,3}),
\]
make \(c_j\) adjacent to the literal vertices corresponding to the three
literal occurrences \(\ell_{j,1},\ell_{j,2},\ell_{j,3}\).  Thus, if
\(\ell_{j,t}=x_i\), then \(c_j\) is adjacent to the literal vertex \(x_i\),
and if \(\ell_{j,t}=\bar{x}_i\), then \(c_j\) is adjacent to the literal
vertex \(\bar{x}_i\).  Since the graph is simple, repeated occurrences of
the same literal create no additional edge.

This completes the construction.  It is clearly computable in polynomial
time.

\paragraph{The degree upper bound.}
Each variable-helper vertex \(h_i\) has degree \(2\), and each
clause-helper vertex \(q_j\) has degree \(2\).  All other vertices have
degree at least \(2\).  Indeed, each anchor is contained in the clique
\(K\), each literal vertex is contained in \(K\) and is also adjacent to
its variable helper, and each clause vertex \(c_j\) is adjacent to \(q_j\)
and to the literal vertices of its clause.  Hence
\[
  \min\deg(G_\varphi)=2.
\]
Therefore, by the standard upper bound
\[
  \dom(G)\leq \min\deg(G)+1,
\]
we obtain
\[
  \dom(G_\varphi)\leq 3.
\]

\paragraph{Completeness.}
Assume that \(\varphi\) is satisfiable, and fix a satisfying truth
assignment.  We construct a \(3\)-domatic coloring of \(G_\varphi\) with
colors
\[
  T,\ F,\ B.
\]
Color the anchor vertices by
\[
  \tau=T,\qquad f=F,\qquad b=B.
\]

For each variable \(x_i\), color the literal vertices and the variable
helper as follows.  If \(x_i\) is true under the chosen assignment, set
\[
  x_i=T,\qquad \bar{x}_i=F,\qquad h_i=B.
\]
If \(x_i\) is false, set
\[
  x_i=F,\qquad \bar{x}_i=T,\qquad h_i=B.
\]
Thus, in either case,
\[
  N[h_i]=\{h_i,x_i,\bar{x}_i\}
\]
sees all three colors.

For each clause \(C_j\), color
\[
  c_j=F,\qquad q_j=B.
\]
Then
\[
  N[q_j]=\{q_j,\tau,c_j\}
\]
has colors \(B,T,F\), and hence \(q_j\) is happy.
Since the assignment satisfies \(C_j\), at least one literal occurrence of
\(C_j\) is true, and hence at least one corresponding literal neighbor of
\(c_j\) has color \(T\).
The vertex \(c_j\) itself has color \(F\),
and \(q_j\) has color \(B\).  Therefore \(N[c_j]\) sees all three
colors.

It remains to check the clique vertices.  Every clique vertex is adjacent
to all other vertices of \(K\), and the clique \(K\) contains the three
anchor vertices \(\tau,f,b\) with colors \(T,F,B\), respectively.
Consequently, every anchor vertex and every literal vertex sees all three
colors in its closed neighborhood.

We have checked every vertex type:
the anchors, the literal vertices, the variable helpers, the clause
helpers, and the clause vertices.  Hence every vertex is happy, and the
three color classes form three disjoint dominating sets of
\(G_\varphi\).  Thus
\[
  \dom(G_\varphi)\geq 3.
\]
Together with \(\dom(G_\varphi)\leq 3\), this gives
\[
  \dom(G_\varphi)=3.
\]

\paragraph{Soundness.}
Conversely, suppose that \(G_\varphi\) has a \(3\)-domatic coloring.
Let \(T\) be the color of \(\tau\).

We first inspect the variable helpers.  For each variable \(x_i\),
\[
  N[h_i]=\{h_i,x_i,\bar{x}_i\}
\]
must see all three colors.  Since this closed neighborhood has exactly three
vertices, \(h_i,x_i,\bar{x}_i\) have three distinct colors.  In particular,
\(x_i\) and \(\bar{x}_i\) cannot both have color \(T\).

Define a truth assignment as follows. If the positive literal vertex \(x_i\)
has color \(T\), set \(x_i\) to true.  If the negative literal vertex
\(\bar{x}_i\) has color \(T\), set \(x_i\) to false.  These two instructions
are never contradictory, because the helper vertex \(h_i\) forces \(x_i\)
and \(\bar{x}_i\) to have distinct colors.  If neither \(x_i\) nor
\(\bar{x}_i\) has color \(T\), assign \(x_i\) arbitrarily.

With this definition, every literal vertex of color \(T\) represents a true
literal under the assignment: a vertex \(x_i\) of color \(T\) makes the
literal \(x_i\) true, and a vertex \(\bar{x}_i\) of color \(T\) makes the
literal \(\bar{x}_i\) true.

Now fix a clause
\[
  C_j=(\ell_{j,1}\vee \ell_{j,2}\vee \ell_{j,3}).
\]
Since
\[
  N[q_j]=\{q_j,\tau,c_j\}
\]
must see all three colors and \(\tau\) has color \(T\), the vertices
\(q_j\) and \(c_j\) must receive the two colors different from \(T\).
In particular,
\[
  q_j \text{ has color different from } T
  \qquad\text{and}\qquad
  c_j \text{ has color different from } T.
\]

The clause vertex \(c_j\) must nevertheless see color \(T\) in its closed
neighborhood.  It does not see color \(T\) from itself, and it does not
see color \(T\) from \(q_j\).
The only remaining neighbors of \(c_j\) are the literal vertices
corresponding to the literal occurrences
\(\ell_{j,1},\ell_{j,2},\ell_{j,3}\).  Therefore at least one of these
literal vertices has color \(T\).

By the definition of the truth assignment, a literal vertex of color
\(T\) corresponds to a true literal.  Hence at least one literal of
\(C_j\) is true.  Since \(C_j\) was arbitrary, every clause of
\(\varphi\) is satisfied.  Thus
\[
  \dom(G_\varphi)\geq 3
  \quad\Longrightarrow\quad
  \varphi\in\SAT.
\]
Equivalently,
\[
  \varphi\notin\SAT
  \quad\Longrightarrow\quad
  \dom(G_\varphi)<3.
\]

Finally, \(G_\varphi\) has no isolated vertices.  Indeed, every anchor and
literal vertex lies in the clique \(K\), every variable helper is adjacent
to its two literal vertices, every clause helper is adjacent to \(\tau\)
and to its clause vertex, and every clause vertex is adjacent to its clause
helper.  Hence, by Lemma~\ref{lem:no-isolated-vertices},
\[
  \dom(G_\varphi)\geq 2.
\]
Therefore, if \(\varphi\notin\SAT\), then
\[
  2\leq \dom(G_\varphi)<3,
\]
and so
\[
  \dom(G_\varphi)=2.
\]

This proves both implications and completes the proof.
\end{proof}

\section{A SAT-to-DNP Reduction Avoiding Domatic Number Three}
\label{sec:main-gap}

We now prove the main reduction theorem.  As in the warm-up construction, we
use the coloring view of domatic partitions: a partition of the vertex set
into \(k\) dominating sets is viewed as a coloring with \(k\) colors, and a
vertex is happy if its closed neighborhood contains all \(k\) colors.

The purpose of this section is to construct the kind of reduction identified
by Riege and Rothe: a reduction from \(\mathrm{3SAT}\) to the domatic-number
problem whose image avoids domatic number \(3\)~\cite{RiegeRothe2004ExactDNP}.
More precisely, we construct a graph \(R(\varphi)\) with domatic number \(4\)
when \(\varphi\) is satisfiable and domatic number \(2\) otherwise.

\begin{theorem}[Main four-versus-two reduction]
\label{thm:main-gap}
There is a polynomial-time computable function \(R\) that maps every
\SAT{} instance \(\varphi\) to a graph \(R(\varphi)\) such that
\[
  \varphi\in\SAT
  \Longrightarrow
  \dom(R(\varphi))=4,
\]
and
\[
  \varphi\notin\SAT
  \Longrightarrow
  \dom(R(\varphi))=2.
\]
\end{theorem}

\begin{proof}
We describe the construction and then prove the two implications.  We may
assume that
\[
  \varphi = C_1\wedge C_2\wedge \cdots \wedge C_m
\]
has \(m\geq 1\) clauses and that each clause contains exactly three literals.
Write the variables as \(x_1,\ldots,x_n\).

\paragraph{Intuition.}
The key trick in the construction is a pigeonhole argument.  We introduce
four anchor vertices, but in a hypothetical \(3\)-domatic coloring these four
anchors receive only three colors.  Hence two anchors must receive the same
color.  Since we do not know in advance which pair this will be, we build
one SAT-checking copy for each unordered pair of anchors.  Thus there are
\(\binom{4}{2}=6\) copies.  In the soundness proof, the pigeonhole principle
selects an anchor pair \(P=\{a,b\}\) whose two anchors have the same color;
the corresponding \(P\)-copy is then forced to encode a satisfying assignment.

In the completeness direction, by contrast, the four anchors are given four
distinct colors.  In the copy indexed by \(P=\{a,b\}\), the two colors of
\(a\) and \(b\) act as the two truth colors for that copy, while the remaining
two colors serve as auxiliary colors.  A satisfying assignment can then be
used to make every vertex happy.

\paragraph{Construction.}
Let
\[
  A=\{a_1,a_2,a_3,a_4\}
\]
be the set of anchor vertices, and let
\[
  \mathcal P=\binom{A}{2}
\]
be the set of unordered anchor pairs.  For \(P=\{a,b\}\in \mathcal P\), we
call \(a\) and \(b\) the anchors of the \(P\)-copy.

For each anchor pair \(P\in\mathcal P\) and each variable \(x_i\), create
four literal vertices
\[
  p_{i,P}^1,\quad p_{i,P}^2,\quad n_{i,P}^1,\quad n_{i,P}^2.
\]
The vertices \(p_{i,P}^1,p_{i,P}^2\) represent the positive literal \(x_i\),
and the vertices \(n_{i,P}^1,n_{i,P}^2\) represent the negative literal
\(\bar{x}_i\).

Put all anchors and all literal vertices into one clique.  Thus the clique
vertices are exactly
\[
  A
  \;\cup\;
  \{p_{i,P}^1,p_{i,P}^2,n_{i,P}^1,n_{i,P}^2
      \mid 1\leq i\leq n,\ P\in\mathcal P\}.
\]

Next, for each anchor pair \(P=\{a,b\}\), each variable \(x_i\), and each
\(\alpha,\beta\in\{1,2\}\), create one consistency helper
\[
  h_{i,P}^{\alpha,\beta}.
\]
This vertex is adjacent exactly to
\[
  a,\quad b,\quad p_{i,P}^{\alpha},\quad n_{i,P}^{\beta}.
\]
No consistency helper is placed in the clique.

Finally, for each anchor pair \(P=\{a,b\}\) and each clause
\[
  C_j=(\ell_{j,1}\vee \ell_{j,2}\vee \ell_{j,3}),
\]
create one clause vertex \(c_{j,P}\) and one clause helper \(q_{j,P}\).
Connect \(q_{j,P}\) exactly to
\[
  a,\quad b,\quad c_{j,P}.
\]
The clause vertex \(c_{j,P}\) is connected to \(q_{j,P}\).  In addition,
for each literal occurrence \(\ell_{j,t}\) of \(C_j\), it is adjacent to the
two literal vertices corresponding to that occurrence.  More explicitly, if
\(\ell_{j,t}=x_i\), then connect \(c_{j,P}\) to
\[
  p_{i,P}^1,\quad p_{i,P}^2,
\]
and if \(\ell_{j,t}=\bar{x}_i\), then connect \(c_{j,P}\) to
\[
  n_{i,P}^1,\quad n_{i,P}^2.
\]
Since the graph is simple, repeated occurrences of the same literal create
no additional edge.

These are all edges of the graph.  This completes the construction of
\(R(\varphi)\).  The construction clearly has size polynomial in
\(|\varphi|\).

The vertex types of \(R(\varphi)\) are therefore the anchors, the literal
vertices, the consistency helpers, the clause vertices, and the clause
helpers.
The anchors and literal vertices are precisely the clique vertices.  The
consistency helpers, clause vertices, and clause helpers are precisely the
vertices outside the clique.

\paragraph{Degree upper bound.}
For every anchor pair \(P=\{a,b\}\) and every clause \(C_j\), the clause
helper \(q_{j,P}\) is adjacent exactly to
\[
  a,\quad b,\quad c_{j,P}.
\]
Hence
\[
  \deg(q_{j,P})=3.
\]
In any \(k\)-domatic coloring, every vertex must see all \(k\) colors in
its closed neighborhood.  Therefore
\[
  k \leq |N[v]|=\deg(v)+1
\]
for every vertex \(v\).  Applying this to \(q_{j,P}\), we obtain
\[
  \dom(R(\varphi))\leq 4.
\]

We also record that \(R(\varphi)\) has no isolated vertices.  The anchors
and literal vertices lie in the big clique, each consistency helper has four
neighbors, each clause helper has three neighbors, and each clause vertex
is adjacent to its clause helper and to literal vertices.  Since every graph
without isolated vertices has domatic number at least \(2\), it follows that
\[
  \dom(R(\varphi))\geq 2.
\]

\paragraph{Completeness.}
Assume that \(\varphi\) is satisfiable, and let \(\sigma\) be a satisfying
truth assignment.  We construct a \(4\)-domatic coloring of \(R(\varphi)\).

Use the colors \(1,2,3,4\).  Color the anchors \(a_1,a_2,a_3,a_4\) with
four distinct colors.  Now fix an anchor pair
\[
  P=\{a,b\}.
\]
Let the colors of \(a\) and \(b\) be the two anchor colors of the \(P\)-copy,
and call the remaining two colors the two background colors of the
\(P\)-copy.

For each variable \(x_i\), color the literal vertices in the \(P\)-copy as
follows.  If \(\sigma(x_i)=\mathrm{true}\), then color
\[
  p_{i,P}^1,\ p_{i,P}^2
\]
with the two anchor colors of \(P\), one each, and color
\[
  n_{i,P}^1,\ n_{i,P}^2
\]
with the two background colors of \(P\), one each.  If
\(\sigma(x_i)=\mathrm{false}\), do the reverse: color
\[
  n_{i,P}^1,\ n_{i,P}^2
\]
with the two anchor colors of \(P\), one each, and color
\[
  p_{i,P}^1,\ p_{i,P}^2
\]
with the two background colors of \(P\), one each.

We next color the vertices outside the clique.

First consider a consistency helper
\[
  h_{i,P}^{\alpha,\beta}.
\]
Its neighbors are
\[
  a,\quad b,\quad p_{i,P}^{\alpha},\quad n_{i,P}^{\beta}.
\]
The vertices \(a\) and \(b\) already give the two anchor colors of the
\(P\)-copy.  Moreover, by the way the literal vertices were colored, exactly
one of \(p_{i,P}^{\alpha}\) and \(n_{i,P}^{\beta}\) has an anchor color, and
the other has a background color.  Let this background color be \(B_1\).
Color \(h_{i,P}^{\alpha,\beta}\) with the other background color \(B_2\).
Then the closed neighborhood
\[
  N[h_{i,P}^{\alpha,\beta}]
\]
contains the two anchor colors, coming from \(a\) and \(b\), and the two
background colors, namely \(B_1\) from one of the literal neighbors and
\(B_2\) from the helper itself.  Hence it contains all four colors.

Now consider a clause \(C_j\) and an anchor pair \(P=\{a,b\}\).  Since
\(\sigma\) satisfies \(C_j\), at least one literal of \(C_j\) is true under
\(\sigma\).  The two literal vertices corresponding to this true literal in
the \(P\)-copy have the two anchor colors of \(P\).  Color
\[
  c_{j,P}
  \quad\text{and}\quad
  q_{j,P}
\]
with the two background colors of \(P\), one each.

We now check every vertex type.

The anchors and literal vertices lie in the big clique.  Since the four
anchors already have the four distinct colors, every clique vertex sees all
four colors in its closed neighborhood.

Every consistency helper has just been colored so that its closed
neighborhood contains the two anchor colors and the two background
colors of its copy.  Hence every consistency helper is happy.

For every clause helper \(q_{j,P}\), we have
\[
  N[q_{j,P}]
  =
  \{q_{j,P},a,b,c_{j,P}\}.
\]
The anchors \(a,b\) have the two anchor colors of \(P\), while
\(q_{j,P}\) and \(c_{j,P}\) have the two background colors of \(P\).
Thus every clause helper is happy.

Finally, consider a clause vertex \(c_{j,P}\).  It sees the two background
colors from itself and \(q_{j,P}\).  It also sees the two anchor colors
from the two literal vertices corresponding to a true literal of \(C_j\).
Thus every clause vertex is happy.

Hence the coloring is a \(4\)-domatic coloring of \(R(\varphi)\).  Therefore
\[
  \dom(R(\varphi))\geq 4.
\]
Together with the degree upper bound \(\dom(R(\varphi))\leq 4\), this gives
\[
  \dom(R(\varphi))=4.
\]

\paragraph{Soundness.}
Assume that \(R(\varphi)\) has a \(3\)-domatic coloring.  We show that
\(\varphi\) is satisfiable.

The four anchors receive only three colors.  Hence, by the pigeonhole
principle, there is an anchor pair
\[
  P=\{a,b\}
\]
whose two anchors have the same color.  Call this color \(T\).  We use
the \(P\)-copy to define a truth assignment.

For each variable \(x_i\), set
\[
  \sigma(x_i)=\mathrm{true}
\]
if at least one of
\[
  p_{i,P}^1,\ p_{i,P}^2
\]
has color \(T\).  If neither \(p_{i,P}^1\) nor \(p_{i,P}^2\) has color
\(T\), set
\[
  \sigma(x_i)=\mathrm{false}.
\]

We first check the consistency helpers.  We claim that no variable has both
a positive and a negative literal vertex of color \(T\) in the \(P\)-copy.
Indeed, suppose that for some \(\alpha,\beta\in\{1,2\}\), both
\[
  p_{i,P}^{\alpha}
  \quad\text{and}\quad
  n_{i,P}^{\beta}
\]
have color \(T\).  The consistency helper \(h_{i,P}^{\alpha,\beta}\) is
adjacent exactly to
\[
  a,\quad b,\quad p_{i,P}^{\alpha},\quad n_{i,P}^{\beta}.
\]
All four of these neighbors have color \(T\).  Hence
\(N[h_{i,P}^{\alpha,\beta}]\) contains at most the color \(T\) and the
color of \(h_{i,P}^{\alpha,\beta}\) itself, and so it contains at most two
colors.  This contradicts the assumption that the coloring is
\(3\)-domatic.  Therefore, if a negative literal vertex \(n_{i,P}^{\beta}\) has
color \(T\), then neither positive literal vertex of \(x_i\) has color
\(T\), and our assignment sets \(\sigma(x_i)=\mathrm{false}\).  Thus the
literal vertices encode truth values consistently.

We next check the clause helpers.  Fix a clause \(C_j\).  Since
\[
  N[q_{j,P}]
  =
  \{q_{j,P},a,b,c_{j,P}\},
\]
and since \(a\) and \(b\) both have color \(T\), the only vertices in
\(N[q_{j,P}]\) that can supply colors different from \(T\) are
\(q_{j,P}\) and \(c_{j,P}\).  Because the coloring is \(3\)-domatic,
\(N[q_{j,P}]\) must contain all three colors.  Hence \(q_{j,P}\) and
\(c_{j,P}\) must receive the two distinct colors different from \(T\).
In particular,
\[
  q_{j,P}\text{ has color different from }T
  \quad\text{and}\quad
  c_{j,P}\text{ has color different from }T.
\]

We now check the clause vertex \(c_{j,P}\).
The vertex \(c_{j,P}\) is adjacent to \(q_{j,P}\) and to the literal vertices
corresponding to the three literal occurrences of \(C_j\).
Since both \(c_{j,P}\) and \(q_{j,P}\) have
colors different from \(T\), the only way for \(c_{j,P}\) to see color
\(T\) in its closed neighborhood is through one of its literal neighbors.
Because the coloring is \(3\)-domatic, \(c_{j,P}\) must see color \(T\).
Therefore at least one literal neighbor of \(c_{j,P}\) has color \(T\).

If this \(T\)-colored literal neighbor is one of
\(p_{i,P}^1,p_{i,P}^2\), then the corresponding literal \(x_i\) is true
under \(\sigma\).  If it is one of \(n_{i,P}^1,n_{i,P}^2\), then, by the
consistency argument above, neither \(p_{i,P}^1\) nor \(p_{i,P}^2\) has
color \(T\), and so \(\sigma(x_i)=\mathrm{false}\); hence the corresponding
literal \(\bar{x}_i\) is true under \(\sigma\).  Thus the clause \(C_j\)
has at least one true literal.

Since \(C_j\) was arbitrary, every clause of \(\varphi\) is satisfied by
\(\sigma\).  Hence
\[
  \varphi\in\SAT.
\]
Thus every \(3\)-domatic coloring of \(R(\varphi)\) yields a satisfying
assignment for \(\varphi\).  By the monotonicity observation stated in the preliminaries, if \(\dom(R(\varphi))\geq 3\), then
\(R(\varphi)\) has a \(3\)-domatic coloring: starting from any
\(r\)-domatic coloring with \(r\geq 3\), merge color classes until only
three colors remain.  Hence
\[
  \dom(R(\varphi))\geq 3
  \quad\Longrightarrow\quad
  \varphi\in\SAT.
\]
Equivalently,
\[
  \varphi\notin\SAT
  \quad\Longrightarrow\quad
  \dom(R(\varphi))<3.
\]
Since \(R(\varphi)\) has no isolated vertices, Lemma~\ref{lem:no-isolated-vertices}
gives
\[
  \dom(R(\varphi))\geq 2.
\]
Therefore, if \(\varphi\notin\SAT\), then
\[
  \dom(R(\varphi))=2.
\]

Combining completeness and soundness proves the theorem.
\end{proof}

\section{Exact Domatic Number Problems}
\label{sec:exact-dnp-application}

As observed in \cref{lem:exact-dnp-in-dp},
\(\ExactDNP{k}\in\DP\) for every fixed \(k\).
This is the standard membership argument for exact optimization problems:
one writes the exact condition as the conjunction of a lower-bound condition
and the complement of the next larger lower-bound condition.  See
\cite{PapadimitriouYannakakis1984} for the class \(\DP\) and exact and
critical problems, and
\cite{RiegeRothe2004ExactDNP} for exact domatic-number
background.

Riege and Rothe proved that
\(\ExactDNP{k}\) is \(\DP\)-complete for every
fixed \(k\geq 5\), and that
\(\ExactDNP{2}\) is \(\coNP\)-complete
\cite{RiegeRothe2004ExactDNP}.  The case \(k=1\) is
polynomial-time decidable: equivalently, one tests whether the input graph has
an isolated vertex.  Thus, before the present work, the only unresolved
single-target exact domatic-number cases were \(k=3\) and \(k=4\). 

Riege and Rothe also pointed out a natural route for closing this remaining
classification gap: it would suffice to find a reduction to the domatic-number
problem whose output graphs all have domatic number different from
\(3\)~\cite{RiegeRothe2004ExactDNP}.  The main reduction theorem of the
present paper provides exactly such a reduction.  Indeed, by
Theorem~\ref{thm:main-gap},
\[
  \dom(R(\varphi)) =
  \begin{cases}
    4, & \text{if \(\varphi\in 3\mathrm{SAT}\),}\\
    2, & \text{if \(\varphi\notin 3\mathrm{SAT}\),}
  \end{cases}
\]
and hence \(\dom(R(\varphi))\in\{2,4\}\) for every formula \(\varphi\).
Together with the warm-up reduction theorem, this gives the short
DP-hardness reduction for \(\mathrm{Exact}\text{-}3\text{-DNP}\) below.

\begin{theorem}[Exact-three domatic number]
\label{thm:exact-three-dnp-dp-complete}
\(\ExactDNP{3}\) is \(\DP\)-complete.
\end{theorem}

\begin{proof}
Membership in \(\DP\) follows from
\cref{lem:exact-dnp-in-dp}.  It remains to prove
\(\DP\)-hardness.

We reduce from the standard \DP-complete problem
\(\SATUNSAT\) \cite{PapadimitriouYannakakis1984}.
We may assume that \(\varphi\) and \(\psi\) are 3-CNF formulas.

Let \(R\) be the reduction from \cref{thm:main-gap}, and let
\(S\) be the reduction from \cref{thm:warm-up-gap}.  Thus
\[
  \dom(R(\varphi))
  =
  \begin{cases}
    4, & \text{if } \varphi\in\SAT,\\
    2, & \text{if } \varphi\notin\SAT,
  \end{cases}
\]
and
\[
  \dom(S(\psi))
  =
  \begin{cases}
    3, & \text{if } \psi\in\SAT,\\
    2, & \text{if } \psi\notin\SAT.
  \end{cases}
\]

Given \((\varphi,\psi)\), construct
\[
  H_3(\varphi,\psi)
  =
  R(\varphi)\dot\cup \bigl(S(\psi)\oplus K_1\bigr).
\]
This construction is polynomial-time computable.  By
\cref{lem:join-clique},
\[
  \dom(S(\psi)\oplus K_1)
  =
  \begin{cases}
    4, & \text{if } \psi\in\SAT,\\
    3, & \text{if } \psi\notin\SAT.
  \end{cases}
\]
By \cref{lem:disjoint-union}, the domatic number of
\(H_3(\varphi,\psi)\) is the minimum of the two displayed values.  Hence we
obtain the following complete case distinction:
\[
\begin{array}{c|c|c}
\text{case} &
\bigl(\dom(R(\varphi)),\dom(S(\psi)\oplus K_1)\bigr) &
\dom(H_3(\varphi,\psi))
\\ \hline
\varphi\in\SAT,\ \psi\notin\SAT & (4,3) & 3 \\
\varphi\in\SAT,\ \psi\in\SAT    & (4,4) & 4 \\
\varphi\notin\SAT,\ \psi\notin\SAT & (2,3) & 2 \\
\varphi\notin\SAT,\ \psi\in\SAT    & (2,4) & 2
\end{array}
\]
Therefore
\[
  \dom(H_3(\varphi,\psi))=3
\]
if and only if
\[
  \varphi\in\SAT
  \quad\text{and}\quad
  \psi\notin\SAT.
\]
Equivalently,
\[
  (\varphi,\psi)\in\SATUNSAT
  \quad\Longleftrightarrow\quad
  H_3(\varphi,\psi)\in \ExactDNP{3}.
\]
Thus \(\ExactDNP{3}\) is \(\DP\)-hard, and hence
\(\DP\)-complete.
\end{proof}


\begin{theorem}[Exact-four domatic number]
\label{thm:exact-four-dnp-dp-complete}
\(\ExactDNP{4}\) is \(\DP\)-complete.
\end{theorem}

\begin{proof}
Membership in \(\DP\) follows from
\cref{lem:exact-dnp-in-dp}.  For hardness, we reduce from
\(\ExactDNP{3}\), which is \(\DP\)-complete by
\cref{thm:exact-three-dnp-dp-complete}.

Given a graph \(G\), construct
\[
  G' = G\oplus K_1.
\]
This construction is polynomial-time computable.  By
\cref{lem:join-clique},
\[
  \dom(G')=\dom(G\oplus K_1)=\dom(G)+1.
\]
Hence
\[
  \dom(G)=3
  \quad\Longleftrightarrow\quad
  \dom(G')=4.
\]
Therefore
\[
  G\in \ExactDNP{3}
  \quad\Longleftrightarrow\quad
  G'\in \ExactDNP{4}.
\]
Thus \(\ExactDNP{4}\) is \(\DP\)-hard, and hence
\(\DP\)-complete.
\end{proof}

\paragraph{Proof of \cref{thm:main-result}.}
The claim follows from
\cref{thm:exact-three-dnp-dp-complete,thm:exact-four-dnp-dp-complete}.

\paragraph{Proof of \cref{cor:exact-k-dnp-classification}.}
For \(k=3\) and \(k=4\), this is \cref{thm:main-result}.  For every
fixed \(k\geq 5\), this is the result of Riege and Rothe
\cite{RiegeRothe2004ExactDNP}.

\begin{remark}
\label[remark]{rem:shift-to-higher-exact-dnp}
The preceding two theorems are the new single-target exact domatic-number
cases.  The higher cases \(k\geq 5\) were already proved by Riege and Rothe
\cite{RiegeRothe2004ExactDNP}.  Conversely, once
\cref{thm:exact-three-dnp-dp-complete} is available, the standard
join-with-a-clique shift gives a short alternative derivation of the hardness
part for every fixed \(k\geq 3\).

Indeed, fix \(k\geq 3\).  Given a graph \(G\), construct
\[
  G_k = G\oplus K_{k-3}.
\]
By \cref{lem:join-clique},
\[
  \dom(G_k)=\dom(G)+k-3.
\]
Therefore
\[
  \dom(G)=3
  \quad\Longleftrightarrow\quad
  \dom(G_k)=k.
\]
Thus the map \(G\mapsto G\oplus K_{k-3}\) reduces
\(\mathrm{Exact}\text{-}3\text{-DNP}\) to
\(\mathrm{Exact}\text{-}k\text{-DNP}\).  Together with the standard membership
in \(\mathrm{DP}\), this gives DP-completeness for every fixed \(k\geq 3\).
For \(k\geq 5\), this recovers the earlier hardness result of Riege and
Rothe~\cite{RiegeRothe2004ExactDNP}.
\end{remark}

\section{Conclusion}

We have closed the two remaining cases in the exact domatic-number classification
left open by Riege and Rothe.  The main technical step is a direct reduction
from \(\mathrm{3SAT}\) to graphs whose domatic numbers are restricted to the
two values \(2\) and \(4\).  In particular, the image of the reduction avoids
domatic number \(3\), which was the missing ingredient for the
exact-domatic-number application.

Together with the earlier results for \(k\geq 5\) and for \(k=2\), this yields
the complete classification for fixed exact domatic-number problems:
\(\mathrm{Exact}\text{-}2\text{-DNP}\) is coNP-complete, while
\(\mathrm{Exact}\text{-}k\text{-DNP}\) is DP-complete for every fixed
\(k\geq 3\).

\bibliographystyle{alphaurl}
\bibliography{references}

\end{document}